\documentclass[10pt, final, journal, letterpaper, twocolumn]{IEEEtran}

\usepackage{cite}
\usepackage{amsmath,amsthm,amssymb,amsfonts}
\usepackage{algorithm}
\usepackage{algorithmic}
\usepackage{graphicx}
\usepackage{color}
\newcommand{\bs}{\boldsymbol}
\usepackage[normalsize]{subfigure}

\newtheorem{proposition}{Proposition}

\newtheorem{corollary}{Corollary}

\begin{document}

\title{Interference-Constrained Pricing for D2D Networks}
\author{Yuan Liu, \IEEEmembership{Member, IEEE}, Rui Wang, \IEEEmembership{Member, IEEE}, and Zhu Han, \IEEEmembership{Fellow, IEEE}
\thanks{This paper was presented in part at the IEEE International Conference on Communications (ICC), Sydney, Australia, June 2014 \cite{YuanICC14}.}
\thanks{Y. Liu is with the School of Electronic and Information Engineering, South China University of Technology, Guangzhou 510641, China (email: eeyliu@scut.edu.cn).}
\thanks{R. Wang is with the Department of Information and Communications at Tongji University, Shanghai 201804, China (email: ruiwang@tongji.edu.cn).}
\thanks{Z. Han is with the Department of Electrical and Computer Engineering, University of Houston, TX 77004, USA (e-mail: zhan2@uh.edu).}

}

\maketitle

\vspace{-1.5cm}

\begin{abstract}
The concept of device-to-device (D2D) communications underlaying cellular networks opens up  potential benefits for improving system performance but also brings new challenges such as interference management. In this paper, we propose a pricing framework for interference management from the D2D users to the cellular system,  where the base station (BS) protects itself (or its serving cellular users) by pricing the cross-tier interference caused from the D2D users. A Stackelberg game is formulated to model the interactions between the BS and D2D users. Specifically, the BS sets prices to a maximize its revenue (or any desired utility) subject to an interference temperature constraint. For given prices, the D2D users competitively adapt their power allocation strategies for individual utility maximization. We first analyze the competition among the D2D users by noncooperative game theory and an iterative based distributed power allocation algorithm is proposed. Then, depending on how much network information the BS knows, we develop two optimal algorithms, one for uniform pricing with limited network information and the other for differentiated pricing with global network information. The uniform pricing algorithm can be implemented by a fully distributed manner and requires minimum information exchange between the BS and D2D users, and the differentiated pricing algorithm is partially distributed and requires no iteration between the BS and D2D users. Then a suboptimal differentiated pricing scheme is proposed to reduce complexity and it can be implemented in a fully distributed fashion.  Extensive simulations are conducted to verify the proposed framework and algorithms.
\end{abstract}

\begin{keywords}
  Device-to-Device (D2D), interference management, distributed power allocation, pricing, and game theory.
\end{keywords}

\section{Introduction}

Incorporating device-to-device (D2D) communications as an underlay to cellular networks attracts considerable interests due to its potential benefits, like spectral/energy efficiency improvements, coverage extension, and traffic offloading. Such a heterogenous network consisting of infrastructure-based and ad hoc networks can achieve better performance than in a pure cellular or ad hoc network \cite{Corson2010,Doppler2009,Lei2012,Fodor2012}. Very recently, D2D in cellular networks has been studied and standardized by the Third Generation Partnership Project (3GPP) Long Term Evolution Advanced (LTE-A) \cite{3gpp1,3gpp2}.

However, the D2D enabled cellular networks pose new challenges which are quite different from those of either cellular networks or ad hoc networks, and thus significantly complicate the network design. One of the most crucial issues is interference management. With spectrum reuse, D2D communication
can improve spectral efficiency and thus enhance system throughput. However, intracell interference could be severe in addition to intercell interference, because D2D communication also causes interference to
the cellular network users. Therefore, methods for efficient interference management and coordination
must be developed to both cellular and D2D users for taking full advantage of D2D communications.
 On the other hand, it is important to guarantee that D2D communications do not generate harmful interference to cellular communications since cellular networks operate on licensed bands. This is similar to cognitive radio systems, and one of the major differences is that D2D communications can be controlled by cellular base station (BS) \cite{Doppler2009,Lei2012,Fodor2012,XiaICCC16}, whereas secondary users are not controlled by primary users in cognitive radio networks.

There are several works for interference management in D2D communications \cite{Zulhasnine2010,Yu2011,Min2011,Xu2013,YuanCL2016}. For instance, authors in \cite{Zulhasnine2010} formulated the channel assignment problem as a mixed integer nonlinear programming and proposed a greedy heuristic algorithm to maximize total throughput while maintaining signal-to-interference-plus-noise ratio (SINR) requirements of all users. In \cite{Yu2011}, power allocation was studied for throughput maximization with minimum and maximum
SINR constraints. Both \cite{Zulhasnine2010} and \cite{Yu2011} assumed that one channel can be occupied by at most one D2D-cellular user pair. Interference management schemes were presented in \cite{Min2011} based on a predefined interference limited area. The authors in \cite{Xu2013} adopted a combinatorial auction approach to assign the cellular users' channels to the D2D users for total throughput maximization. Optimal centralized and distributed mode selection between cellular communication and D2D communication for each user were investigated in \cite{YuanCL2016}. Joint channel and power allocation using in combinatorial auction was also studied in \cite{6981957,6845058} .

In a realistic spectrum-sharing network, the interests of the D2D-tier and the cellular-tier may be inconsistent due to the cross-tier (i.e., D2D-to-cellular and cellular-to-D2D) interference. To this end, in this paper, we target at jointly optimizing the possibly conflicting objectives of the two tiers, which is essentially different from the above mentioned works \cite{Zulhasnine2010,Yu2011,Min2011,Xu2013} which are system-wide optimization by resource allocations.

Specifically, to ensure that the aggregate received interference at the BS is kept below an acceptable level,
we impose an \emph{interference temperature constraint} at the BS, and the BS prices the received interference caused from the D2D users.
%
%
We note that interference temperature constraint has been commonly adopted at BS in heterogeneous networks (e.g., \cite{Kang2012,Ma2010}) and primary user in cognitive radio networks (e.g., \cite{Kim2011,Pang2010a,Wu2009,Hong2011}). Nevertheless, in the schemes \cite{Ma2010,Kim2011,Pang2010a,Wu2009,Hong2011}, the incentives of BS and primary user were not considered. In other words, the BS and primary user's prices are used to maintain the interference temperature constraint but not to maximize the utilities of themselves, i.e. the BS and primary user have no utility.
Though \cite{Kang2012} considered the BS's incentive, the authors mainly focused on sparse-deployment users, i.e., the mutual interferences among users can be neglected.

Unlike these  previous works, in this paper, by imposing the interference temperature constraint at the BS side, the interference tolerance margin at the BS is treated as a divisible resource to be sold among the D2D users.
This is because enabling D2D communications underlaying cellular networks brings severe interference to the original cellular system due to spectrum-sharing. Thus D2D users pay extra prices (costs) for causing undesirable interferences to the cellular system. This is the idea of interference pricing.
Moreover, interference pricing is used not only as a game-theoretic approach to balance the objectives of the D2D-tier and cellular-tier, but also as a mediator among the mutually interfered D2D users.
These aspects make our pricing framework distinctly differ from the related works on interference temperature constraint in heterogeneous networks and cognitive radio networks (e.g., \cite{Ma2010,Kang2012,Kim2011,Pang2010a,Wu2009,Hong2011,Lin2013}).

It is worth noting that it hardly solves the resource allocation problems in interference channel with globally optimal solutions in general even in a centralized environment. On the other hand, game theory offers a set of mathematical tools to study complex interactions among rational players and adapt their choices of strategies. Therefore, game theory is a suitable tool to model and analyze the resource allocation problems for D2D networks. 
In game theory, the \emph{fictitious} prices are usually used as interaction information to coordinate and control the transmissions of network nodes. In other words, the prices have the economic interpretations but are actually system parameters designed in resource allocation schemes.

In this paper, we model the interactions between the cellular-tier and D2D-tier as a Stackelberg game. In D2D networks, the BS provides services and the D2D users are controlled by the BS for interference management. Thus the relationship of the BS and D2D users is a bit like the hierarchical structure (i.e. leader and followers) of Stackelberg game. Moreover, the nodes may be myopic and maximize their own profits through competition. This is our motivation of using Stackelberg game. As the leader, the BS sells interference to maximize its revenue (or any desired utility) under a maximum interference tolerance margin. As the followers, the D2D users then purchase interference from the BS to maximize their payoffs. More specifically, the D2D users are modeled as selfish players and form a noncooperative power control subgame, where each D2D user optimally chooses its own transmit power based on local channel state information (CSI) in response to the power allocation strategies of the other D2D users. We propose an iteratively  distributed power allocation algorithm to achieve the unique Nash equilibrium (NE) point. 
Such a user-level subgame is a classical noncooperative Nash power game. The merit is that, given any price, the power outputs converge to a stable solution with a distributed fashion. This is important in D2D networks because D2D users communicate with each other by ad-hoc manner and they are usually self-organized. We also note that NE often leads to network performance degradation compared with a globally optimal solution.
Using the idea of pricing in Stackelberg game, we propose two optimal algorithms for uniform and differentiated pricing. We further propose a suboptimal differentiated pricing scheme with closed-form to reduce complexity.
We show that the proposed uniform pricing algorithm can be implemented by a fully distributed manner and requires minimum information exchange between the BS and D2D users,  the proposed optimal differentiated pricing algorithm is partially distributed and requires no iteration between the BS and D2D users, and the proposed suboptimal differentiated pricing algorithm is fully distributed and without iteration.

Note that the basic idea of pricing used in our paper is common with \cite{Kang2012}, but the proposed optimal algorithms for uniform and differentiated pricing are new, which is the main contribution of our paper. 
Specifically, the work \cite{Kang2012} mainly focused on sparse deployment (i.e. without interference among users) and solved the problem by using Lagrangian duality method for both uniform and differentiated pricing. For dense deployment (i.e. with interference among users), the authors in \cite{Kang2012} solved the problem by exhaustive search over all feasible spaces for both uniform and differentiated pricing.
In our paper, we focus on dense deployment because it is more general in practice. For the uniform pricing, we shrink the search space of price into a specific range by analyzing the properties of the problem, which significantly reduces the complexity. We further reduce the search complexity by exploring the structure of the power allocation.  
For the differentiated pricing, we first express the revenue function of the BS in terms of transmit powers at the Nash equilibrium point of the user-level subgame. Then the revenue maximization problem is transformed into a linear programming problem. By solving the linear programming problem, we can get the optimal prices and, in turn, enforce the users to transmit the desired powers at the Nash equilibrium point of the user-level subgame.
We further propose a suboptimal algorithm with closed-form for differentiated pricing to reduce complexity, which can be implemented in a fully distributed manner. 
Therefore, we conclude that the proposed three algorithms are fundamentally different from the exhaustive search method in \cite{Kang2012}.

The remainder of this paper is organized as follows. Section II describes the system model and game-theoretic problem formulation. In Section III, a distributed power allocation algorithm for the D2D users as well as pricing algorithms are developed. Comprehensive simulations are provided in Section IV. Finally, we conclude the paper in Section V.

\section{System Model and Problem Formulation}

\subsection{System Model}

\begin{figure}[t]
\begin{centering}
\includegraphics[scale=1.]{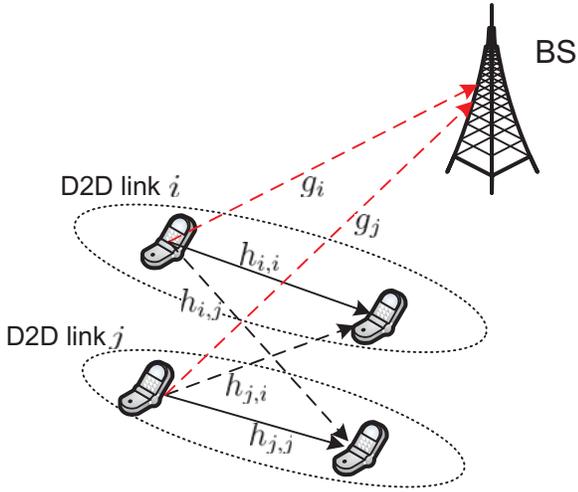}
\vspace{-0.1cm}
 \caption{System model of D2D communication underlaying a cellular network, where the solid and dotted lines denote desired and interference signals, respectively.}\label{fig:system}
\end{centering}
\vspace{-0.3cm}
\end{figure}

We consider a single two-tier cellular network as shown in Fig. \ref{fig:system}, where  total $N$ D2D users share the same \emph{uplink} spectrum of the cellular system and are allowed to transmit simultaneously, i.e. frequency reuse factor of 1. When a D2D source transmits information to its dedicated destination, it not only harms other D2D destinations (co-tier interference) but also interferes with the BS (cross-tier interference). This paper focuses on controlling the cross-tier interference, i.e., the interference from D2D sources to the BS, since it needs to protect the original cellular system if integrating D2D communications.
Denote the D2D users as a set $\mathcal N:=\{1,\cdots,N\}$ and each D2D user refers to a source-destination pair.
In what follows, we use ``user(s)" instead of ``D2D user(s)" for convenience.

We assume that the transmission is based on slot basis, the fading remains unchanged during each transmission slot but possibly varies from one slot to another.
Each slot is assumed to be divided into two phases: the signaling phase and the transmission phase. The signaling phase is used for information interaction (or iteration until convergence), and the transmission phase is used for data transmission with constant power strategies fixed in the signaling phase. Here it is assumed that a slot can be designed long enough in the system so that the overhead of the signaling phase is negligible.
As shown in Fig. \ref{fig:system}, the channel gain from source $i$ to destination $j$ is denoted by $h_{i,j}$, and the source $i$ between the BS is denoted by $g_i$.
All channel gains are modeled as large-scale path loss along with small-scale Rayleigh fading.
The additive noise at the destination of user $i$ is assumed to be independent circularly symmetric complex
Gaussian random variables with zero mean and variance $\sigma^2$. The transmit power of user $i$ is denoted by $p_i$, and denote $\bs p:=\{p_1,\cdots,p_N\}$.

\subsection{Stackelberg Game Formulation}

\begin{figure}[t]
\begin{centering}
\includegraphics[scale=1.2]{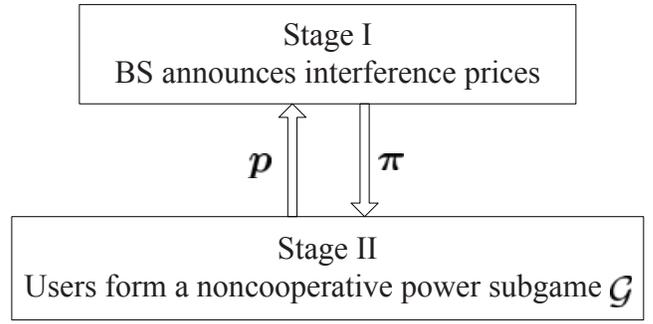}
\vspace{-0.1cm}
 \caption{The proposed Stackelberg game.}\label{fig:game}
\end{centering}
\vspace{-0.3cm}
\end{figure}

We consider the Stackelberg game between the BS and the users as shown in Fig. \ref{fig:game}. The BS is the Stackelberg leader, it first measures the interference temperature and then sets the interference prices in Stage I. The users are followers and choose their actions of power allocation to maximize their individual payoffs in Stage II, according to the prices announced from the BS in Stage I.
We assume that the selfish players are myopic, which means that the players maximize their immediate expected payoffs and do not intentionally affect the strategies of others.

There are a variety of pricing schemes for wireless resource allocations. For instance, individual users are charged in proportion to transmit power \cite{Saraydar2001,Xiao2003,Wang2009,YuanAuction2013}, spectrum trading/access \cite{Niyato2009,Niyato2008}, received SINR \cite{Huang2008,Rasti2009,Ren2011}, and throughput \cite{Feng2004,Kelly1997,Kelly1998}. For the problem considered in this paper, the BS must guarantee that the aggregate received interference from all users should be below the predefined threshold. Hence, it is reasonable that the BS charges each user $i$ by a price $\pi_i$ corresponding to its caused interference $p_ig_i$ to the BS.

Denote $\bs\pi:=\{\pi_1,\cdots,\pi_N\}$. Here we treat the BS as the Stackelberg leader. The goal of the BS is to set the optimal interference prices $\bs\pi$ to maximize its revenue charged from the users within its tolerable aggregate interference margin. Mathematically, the optimization problem at the BS's side can be expressed as
\begin{subequations}\label{eqn:bs}
\begin{align}
\textbf{P1}:~~\max_{\bs\pi\succeq0}~ &u_B(\bs\pi)=\sum_{i=1}^N p_ig_i\pi_i\\
{\rm s.t.}~~ &\sum_{i=1}^N p_ig_i\leq I_{th}, \label{eqn:ith}
\end{align}
\end{subequations}
where $I_{th}$ is the maximum interference that the BS can tolerate, and \eqref{eqn:ith} is the interference power constraint or interference temperature constraint, which means that the total received power from D2D users at the BS should be below a threshold.


%

At the users' side, the received SINR of user $i$ can be written as
\begin{equation}
  \gamma_i(p_i,p_{-i})=\frac{p_ih_{i,i}}{\sum_{j\neq i}p_jh_{j,i}+\sigma^2},
\end{equation}
where $p_{-i}:=\{p_1,\cdots,p_{i-1},p_{i+1},\cdots,p_N\}$ is the vector of power allocation of all the users except for user $i$.

The achievable rate of user $i$ is given by
\begin{equation}
  R_i(\gamma_i(p_i,p_{-i}))=\log(1+\gamma_i(p_i,p_{-i})).
\end{equation}

We define the payoff of user $i$ as
\begin{equation}\label{eqn:ui}
  u_i(p_i,p_{-i},\pi_i)=w_iR_i(\gamma_i(p_i,p_{-i}))-p_ig_i\pi_i,
\end{equation}
where $w_i$ is the weight of user $i$. The payoff function of each user is the difference between its transmission rate and the payment that it needs to make to the BS for the caused interference. Note that the proposed algorithms in this paper are not affected if the transmission rates are replaced by general utility functions that are differentiable, strictly increasing, and concave. We consider the transmission rates as the utility functions only for ease of exposition.

At the users' side, each user $i$ aims to maximize its own payoff by power adaption for given price $\pi_i$ set by the BS, this problem can be formulated as
\begin{subequations}\label{eqn:ui-game}
\begin{align}
\textbf{P2}:~~\max_{p_i}~ &u_i(p_i,p_{-i},\pi_i) \\
{\rm s.t.}~~ &0\leq p_i\leq \overline{p}_i.
\end{align}
\end{subequations}

Note that the BS's prices and the users' power allocation strategies are coupled in a very sophisticated way. Specifically, the BS's pricing decisions influence the users' power allocation strategies which, in turn, impact the BS's revenue. In the following, we propose a Stackelberg game approach to study their interactions.

\section{User and Base Station Optimization}

The Stackelberg game falls into the class of dynamic game and the objective of the game is to find the Stackelberg Equilibrium (SE) point(s) which can be obtained by finding its subgame perfect Nash equilibrium (NE). The typical solution for determining SE is backward induction \cite{Colell}. Therefore, we will start the Stackelberg game by analyzing the users' behaviors in Stage II given the BS's pricing decisions. Then we will investigate the BS's pricing strategies considering the interference temperature constraint in Stage I. In this stage, we propose two optimal pricing schemes, i.e., uniform pricing and differentiated pricing. Notice that the backward induction captures the sequential dependence of the decisions in the two stages. Then a suboptimal differentiated pricing algorithm is further proposed to reduce complexity. At the end of this section, we simply discuss the complexity of the proposed algorithms.

\subsection{Distributed Power Allocation}

Noncooperative game-theoretic approaches are effective to characterize the selfish behaviors of self-interested players. Knowing the prices set by the BS, the competition among the users can be mathematically formulated as a noncooperative power control subgame
\begin{equation}
  \mathcal G:=\{\mathcal N,\{\mathcal P_i\}, \{u_i\}\},
\end{equation}
where $\mathcal N$ is the set of players (or users), $\mathcal P_i$ is the strategy space of each user $i\in\mathcal N$ and defined as the interval $\mathcal P_i:=\{p_i,0\leq p_i\leq \overline{p}_i\}$ that contains the power allocation choices, and $u_i$ is the payoff of each user $i\in\mathcal N$ defined in \eqref{eqn:ui}.

The common concept for solving the noncooperative game problems is the NE at which no user can increase its payoff by unilaterally changing its own transmit power. Mathematically, the power profile $\bs p^*=\{p_1^*,p_2^*,\cdots,p_N^*\}$ is the NE point for the user-level subgame $\mathcal G$ if, for every user, $u_i(p_i^*, p_{-i}^*, \pi_i)\geq u_i(p_i, p_{-i}^*, \pi_i)$, $\forall p_i\in\mathcal P_i$, $\forall i\in\mathcal N$.

Another common concept in game theory is the \emph{best response function} $\mathcal B_i(p_{-i})$ for each player. Formally, define $\mathcal P_{-i}:=\prod_{j\neq i}\mathcal P_j$ as the set-valued function that assigns the best powers to each interference power vector $p_{-i}\in\mathcal P_{-i}$, then $\mathcal B_i(p_{-i}):=\{p_i\in\mathcal P_i|u_i(p_i, p_{-i}, \pi_i)\geq u_i(p_i', p_{-i}, \pi_i),\forall p_i'\in\mathcal P_i\}$.
The best response function $\mathcal B_i(p_{-i})$ reflects the best power user $i$ should transmit in response to the other users' power strategies for the given price set by the BS.

It is easy to verify that the objective function of \textbf{P2} in \eqref{eqn:ui-game} is concave in $p_i$, and the constraint is affine. Thus \textbf{P2} is a convex problem and its optimal solution must satisfy the Karush-Kuhn-Tucker (KKT) conditions \cite{Boyd}. 
By taking the partial derivative of $u_i(p_i,p_{-i},\pi_i)$ with respect to $p_i$ and equating the result to zero,
the best response function $\mathcal B_i(p_{-i})$ can be derived as the following closed-form:
\begin{equation}\label{eqn:best}
  \mathcal B_i(p_{-i})=\left[\frac{w_i}{g_i\pi_i}-\frac{\Delta_i(p_{-i})}{h_{i,i}}\right]_0^{\overline{p}_i},
\end{equation}
where $[x]_a^b:=\max\{\min\{x,b\},a\}$ and $\Delta_i(p_{-i}):=\sum_{j\neq i}p_jh_{j,i}+\sigma^2$ is the interference-plus-noise (IpN) term.

Denote $\mathcal B(\bs p):=\{\mathcal B_1(p_{-1}),\mathcal B_2(p_{-2}),\cdots,\mathcal B_N(p_{-N})\}$, then we present an iterative distributed algorithm for the noncooperative power allocation subgame $\mathcal G$ in Algorithm 1.
\begin{algorithm}[!t]
\caption{Iterative Distributed Power Allocation}
\begin{algorithmic}[1]
\STATE Given price vector $\bs\pi\succeq0$. \STATE Set $t=0$ and initialize $\bs p^{(0)}$ as any feasible vector.
\REPEAT \STATE $t\leftarrow t+1$; \STATE $\bs p^{(t+1)}\leftarrow \mathcal B\left(\bs p^{(t)}\right)$. \UNTIL{$\bs p$ converges.}
\end{algorithmic}
\end{algorithm}

At the beginning of Algorithm 1, each user $i$ arbitrarily chooses its initial power level in its own strategy space $\mathcal P_i$ and then $\bs p^{(0)}$ is feasible.

For a noncooperative game, it is of great important to study the existence of the pure NE point
and the convergence of the iterative process, which are critical for the outcome of the game being stable and eventually arriving. Before leaving this subsection, we give the following proposition to guarantee the existence and uniqueness of the NE point in the user-level subgame and the convergence of the proposed Algorithm 1.

\begin{proposition}\label{prop:ne}
For any given price vector $\bs\pi\succeq0$, the pure NE point of the user-level subgame $\mathcal G$ exists and is unique. Moreover, Algorithm 1 always converges to the unique NE point for any initial feasible power vector $\bs p^{(0)}$.
\end{proposition}

\begin{proof}
  Please see Appendix \ref{app:ne}.
\end{proof}

It is worth noting that the existence of a fixed point (even it is unique) of an iterative process does not necessarily maintain the convergence, and the existence of a fixed point and convergence are two separate concepts of an iterative process. We can prove that the proposed iterative distributed algorithm in Algorithm 1 can converge to the unique NE since the best response function is standard and each user has a peak power constraint \cite{Yates}.



\subsection{Interference Pricing}

Now we turn to look at how the BS makes the pricing decisions for revenue maximization with interference constraint in Stage I. Finding the optimal solution of \textbf{P1} in \eqref{eqn:bs} always resorts to exhaustive search over all ranges of $\bs\pi\succeq\bs0$. Due to the prohibitively computational complexity of the exhaustive search, we alternatively propose two efficient pricing algorithms, one for uniform pricing and the other for differentiated pricing with limited and global network information, respectively. The uniform pricing algorithm assigns an identical price to all users, while the differentiated pricing algorithm charges different received interference power levels by different prices.

\subsubsection{Uniform Pricing with Limited Information}

In this case, the BS sets and broadcasts a uniform price to all users, i.e., $\pi_1=\pi_2\cdots=\pi_N=\pi$. Then the interference-constrained revenue maximization problem in \textbf{P1} reduces to a one-dimensional search problem over price $\pi\geq0$. To demonstrate more insights into the interaction between the BS and users, we first analyze the properties of the revenue function $u_B(\pi)$ in the following proposition.

\begin{proposition}\label{prop:ub}
The optimal uniform price $\pi^*$ must satisfy $\pi^l\leq\pi^*\leq\pi^u$, where $\pi^u$ and $\pi^l$ are the upper and lower bounds of $\pi^*$
\begin{align}
  &\pi^u=\max_{i\in\mathcal N}\frac{w_ih_{i,i}}{g_i\sigma^2},\label{eqn:pi-u}\\
  &\pi^l=\min_{i\in\mathcal N}\frac{w_ih_{i,i}}{g_i\left(\overline{p}_ih_{i,i}+\Delta_i(\overline{p}_{-i})\right)}.\label{eqn:pi-l}
\end{align}
The BS's revenue function $u_B(\pi)$ has the following properties:
\begin{enumerate}
  \item $u_B(\pi)\geq0$;
  \item $u_B(\pi)<\infty$ if the number of users is finite;
  \item $u_B(\pi)=0$ if and only if $\pi=0$ or $\pi\geq\pi^u$;
  \item $u_B(\pi)=\pi\sum_{i=1}^N\overline{p}_i g_i$ if and only if $0\leq\pi\leq\pi^l$.
\end{enumerate}
\end{proposition}

\begin{proof}
Please see Appendix \ref{app:ub}.
\end{proof}

The above proposition yields the following interpretation: the optimal price lies in a certain range, depending on channel conditions, user weights, interference, and peak power constraints; the BS's revenue is always nonnegative because the transmit powers of the users are nonnegative; the maximum revenue is bounded if the number of users is finite; the revenue vanishes if the price is too high or too low.

According to Proposition \ref{prop:ub}, we derive the following corollary.

\begin{corollary}\label{coro:pi}
  When $\pi>\pi^l$, the peak power constraint of each user is not active at the NE point in the user-level subgame $\mathcal G$. In this case, $0\leq \mathcal B_i(p_{-i})<\overline{p}_i$, $\forall i\in\mathcal N$.
\end{corollary}

\begin{proof}
Please see Appendix \ref{app:pi}.
\end{proof}




Based on the above analysis, we know that, 1) If $0\leq\pi\leq\pi^l$, each user transmits its maximum power and the aggregate received interference at the BS is upper bounded, and the corresponding payment to the BS is linear with the price. The intuitive explanation is that, if the BS's price is low enough, every user can afford the payment charged by the BS and will transmit power at a high level; 2) When $\pi\geq\pi^l$, every user reduces its transmit power due to the increased payment charged by the BS, and the transmit power of each user is decreasing with the price; 3) The BS's revenue finally vanishes if $\pi\geq\pi^u$.

Note that the monotonicity of $u_B(\pi)$ between the interval $[\pi^l,\pi^u]$ is very complex. However, the above analysis significantly reduces the complexity.
By exploiting these observations, we propose a feasible way to find the optimal price: We first divide the price interval $[\pi^l,\pi^u]$ into sufficiently small intervals, and for each small interval, the BS picks a price that falls into the small interval and estimates the aggregate received interference. The BS bargains with the users over all candidate prices and finally chooses the price that results in the maximum revenue while maintaining the interference temperature constraint.
Formally, we present the distributed uniform pricing algorithm in Algorithm 2.
\begin{algorithm}[!t]
\caption{Uniform Pricing for interference management}
\begin{algorithmic}[1]
\STATE The BS initializes the interference price as $\pi^u$. \STATE $\tau=0$ and $\epsilon$ is a small positive constant.
\REPEAT \STATE $\tau\leftarrow\tau+1$; \STATE Every user runs Algorithm 1; \STATE The BS measures the total received interference and computes revenue; \IF{$\sum_{i=1}^Np_ig_i\leq I_{th}$} \STATE $\pi^{(\tau+1)}\leftarrow\pi^{(\tau)}-\epsilon$; \ELSE \STATE $\pi^*\leftarrow\pi^{(\tau)}$. \STATE \textbf{break;} \ENDIF \UNTIL{$\pi^{(\tau+1)}\leq \pi^l$.} \STATE Output $\pi^*\leftarrow\arg\max_{\pi^{(\tau)}} u_B(\pi^{(\tau)})$.
\end{algorithmic}
\end{algorithm}

There are several points need to be noted. First, in Algorithm 2, the initial price is selected as $\pi^u$ rather than $\pi^l$.
%
If choose $\pi^l$ as the initial price, it only requires to slightly modify Algorithm 2 and the details are omitted here for brevity.

Second, in the process of price bargaining, the BS gradually decreases the price with the step size $\epsilon$ as the initial price is $\pi^u$.
The algorithm ends as long as the interference temperature constraint is active and there is no need for bargaining the rest of price candidates. This is due to the fact that the transmit power of each user is decreasing with the price in $[\pi^l,\pi^u]$, and so is the the total received interference at the BS, which have been proved in the proof of Corollary \ref{coro:pi}.

Third, it is easy to calculate the revenue. In the proposed uniform pricing framework, the revenue is the product of the uniform price and the total received interference at the BS. Therefore, the BS can obtain its revenue directly after measuring the total received interference. One also observes that, in Algorithm 2, the message passing between the BS and the users is only the broadcasted value of the price in each iteration.

\subsubsection{Differentiated Pricing with Global Information}

Here we consider the general case where different users are charged by different prices at the BS, i.e., differentiated pricing (also known as price discrimination in economics), by assuming that the BS knows the global network information.
%

Denote $\bs\pi=\{\pi_1,\cdots,\pi_N\}$ and $\bs p^*=\{p_1^*,\cdots,p_N^*\}$ the corresponding optimal power vector on the NE point of the user game. By reorganizing \eqref{eqn:best}, the optimal price $\bs\pi$ and the optimal power vector on the NE point $\bs p^*$ have the following relationship:
  \begin{equation}\label{eqn:pi-diff}
    \pi_i=\frac{w_ih_{i,i}}{g_i\left(\sum_{j=1}^Np_j^* h_{j,i}+\sigma^2\right)},~~ \forall i\in\mathcal N.
  \end{equation}

Using \eqref{eqn:pi-diff}, we can express the revenue function of the BS in terms of transmit powers at the NE point of the user-level subgame. As a consequence, the interference-constrained revenue maximization problem in \textbf{P1} can be rewritten as
\begin{subequations}\label{eqn:transfer}
\begin{align}
\max_{\bs p^*}~&u_B(\bs\pi(\bs p^*))=\sum_{i=1}^N\frac{p_i^*w_ih_{i,i}}{\sum_{j=1}^Np_j^* h_{j,i}+\sigma^2} \\
{\rm s.t.}~~&\sum_{i=1}^N p_i^*g_i\leq I_{th}\\
&0\leq p_i^*\leq \overline{p}_i,\forall i\in\mathcal N.
\end{align}
\end{subequations}


By letting 
\begin{equation*}
y_i=\frac{p_i}{\sum_{j=1}^Np_j^* h_{j,i}+\sigma^2}
\end{equation*}
and
\begin{equation*}
z_i=\frac{1}{\sum_{j=1}^Np_j^* h_{j,i}+\sigma^2},
\end{equation*}
%
the problem in \eqref{eqn:transfer} can be transformed as
\begin{subequations}\label{eqn:linear}
\begin{align}
\max_{\{y_i,z_i\}}~&\sum_{i=1}^N w_ih_{i,i}y_i \\
{\rm s.t.}~~&\sum_{i=1}^N g_iy_i - I_{th}z_i\leq0,\forall i\in\mathcal N\\
&y_i - \overline{p}_iz_i\leq0,\forall i\in\mathcal N\\
&\sum_{j=1}^N h_{j,i}y_i + z_i\sigma^2=1,\forall i\in\mathcal N\\
& y_i\geq0,z_i\geq0,\forall i\in\mathcal N.
\end{align}
\end{subequations}

By doing so, it can be easily shown that the problem in \eqref{eqn:linear} is a standard linear programming problem \cite{Boyd}, and thus a global optimum can be computed very efficiently.\footnote{Many numerical solvers for linear programming are available, e.g., MOSEK and CVX.}
After finding the optimal $y_i^*$ and $z_i^*$ by solving \eqref{eqn:linear}, $\bs p^*$ can be recovered using $p_i^*:=y_i^*/z_i^*$, $\forall i$, and then $\bs\pi^*$ can be determined using \eqref{eqn:pi-diff}.

Note that \eqref{eqn:pi-diff}-\eqref{eqn:linear} are used for determining the optimal prices $\boldsymbol\pi$. Though the optimal power $\bs p^*$ can also obtained in this process, it is exactly the same as the NE point of user-level game $\bs p^*(\boldsymbol\pi)$ for given optimal price $\boldsymbol\pi$:
the BS broadcasts the prices $\bs\pi$ to the users and, in turn, enforces the users to transmit the desired powers $\bs p^*$ at the NE point of the user-level subgame. This is true since we have proved in Proposition \ref{prop:ne} that the user-level subgame $\mathcal G$ always converges to the unique NE point for any given price vector.

Finally, we summarize the differentiated pricing algorithm in Algorithm 3. Note that the proposed Algorithm 3 does not need to iterate between the BS and users. Nevertheless, the BS needs to collect the global network information from the users for computing prices. This is possible since D2D users are controlled by the BS \cite{Doppler2009,Lei2012,Fodor2012}.
%
%
We also need to point out that,  the BS does not need to intelligently tell or control the users which power strategies to make, that is, the power allocations of the users still remain a distributed fashion with local CSI. Hence, Algorithm 3 is a \emph{partially} distributed algorithm.

\begin{algorithm}[!t]
\caption{Differentiated Pricing for interference management}
\begin{algorithmic}[1]
\STATE Given any price vector $\bs\pi$, every user runs Algorithm 1.
\STATE The BS collects the global network information for solving the problem in \eqref{eqn:linear}, and then finds the optimal price vector $\bs\pi$ using \eqref{eqn:pi-diff}.
\end{algorithmic}
\end{algorithm}

\subsubsection{Suboptimal Differentiated Pricing}
Algorithm 3 finds the optimal differentiated pricing policy, but it needs to be computed numerically when solving \eqref{eqn:linear}. Here we propose a suboptimal differentiated pricing scheme which has closed-form and thus significantly reduces the computational complexity. The suboptimal scheme is based on two assumptions: (i) See \eqref{eqn:ith}, stronger $g_i$ yields to higher utility from a seller's perspective. Thus the BS may just pre-set the interference tolerance margin for all users in proportion to $\{g_i\}$, then the interference temperature constraint in \eqref{eqn:ith} can be written as
\begin{equation}\label{eqn:subopt-th}
  p_ig_i\leq\frac{g_i}{\sum_{i\in\mathcal N}g_i}I_{th},\forall i\in\mathcal N.
\end{equation}
(ii) The transmitter-receiver distance of a D2D connection is usually very short, thus we assume that $h_{i,i}\gg h_{j,i}$  ($\forall j\neq i$) due to the effects of path-loss when the D2D users are uniformly distributed. In this case, the IpN term $\Delta_i(p_{-i})=\sum_{j\neq i}p_jh_{j,i}+\sigma^2\approx\sigma^2$, $\forall i$, and the optimal power of user $i$ can be approximated as
\begin{equation}\label{eqn:subopt-p}
  p_i^*\approx\left[\frac{w_i}{g_i\pi_i}-\frac{\sigma^2}{h_{i,i}}\right]_0^{\overline{p}_i}.
\end{equation}

Substituting \eqref{eqn:subopt-th} and \eqref{eqn:subopt-p} into \textbf{P1}, the problem can be decoupled to $N$ parallel subproblems and each having an identical structure:
\begin{subequations}\label{eqn:subopt}
\begin{align}
\max_{\pi_i}~& w_i-\frac{\sigma^2 g_i\pi_i}{h_{i,i}}\\
{\rm s.t.}~~& \left[\frac{w_i}{g_i\pi_i}-\frac{\sigma^2}{h_{i,i}}\right]\leq\frac{ I_{th}}{\sum_{i\in\mathcal N}g_i}\\
& \frac{w_ih_{i,i}}{g_i(\overline{p}_ih_{i,i}+\sigma^2)}\leq
\pi_i\leq\frac{w_ih_{i,i}}{g_i\sigma^2},\label{eqn:pii}
\end{align}
\end{subequations}
where \eqref{eqn:pii} ensures $0\leq p_i^*\leq\overline{p}_i$. The optimal price $\pi_i^*$ can be obtained as
\begin{eqnarray}
  \pi_i^*=\begin{cases}\frac{w_ih_{i,i}}{g_i(\overline{p}_ih_{i,i}+\sigma^2)},&{\rm if}~I_{th}\geq\overline{p}_i\sum_{i\in\mathcal N}g_i\\
  \frac{w_ih_{i,i}}{g_i\left(\frac{I_{th}h_{i,i}}{\sum_{i\in\mathcal N}g_i}+\sigma^2\right)},&{\rm otherwise}.
  \end{cases}
\end{eqnarray}

One can see that the suboptimal differentiated pricing scheme only needs the limited network information to compute optimal prices, and the BS can interact each user independently. Thus the suboptimal differentiated pricing scheme is a fully distributed scheme.

\subsection{Complexity Analysis}

In this subsection, we discuss the complexity of the proposed three algorithms at the user side and BS side, respectively.

\subsubsection{Complexity at User Side} For the optimal uniform and differentiated pricing, the needed network information at the user side are the same because the power allocation of the users form the noncooperative subgame $\mathcal G$ by the best response function  $\mathcal B_i(p_{-i})$ in \eqref{eqn:best}. 
Note that each user $i$ knows its weight $w_i$ and $\pi_i$ is broadcasted by the BS, $\mathcal B_i(p_{-i})$ can be computed by user $i$ based on local information only, including $h_{i,i}$, $g_i$, and $\Delta_i(p_{-i})$. Specifically, $h_{i,i}$ and $g_i$ can be obtained via training. That is, the source of user $i$ sends training signals and other nodes receive them, then its own destination and BS estimate $h_{i,i}$ and $g_i$, respectively, where $h_{i,i}$ can be sent via feedback channel, and the uplink channel information $g_i$ can be sent through feedback channel or downlink via time-division duplex (TDD) mode. In addition, the destination of user $i$ can measure the IpN term $\Delta_i(p_{-i})$ and send the value to its source.\footnote{The destination can measure the total received power via reference signals (defined in 3GPP LTE) and is also aware of the desired power from its own source. By extracting the desired received power from the total received power, $\Delta_i(p_{-i})$ is obtained.} 

In summary, at each D2D transmitter $i$, only $h_{i,i}$ and $\Delta_i(p_{-i})$ are needed, which can be measured at its receiver $i$ and fed back to transmitter $i$. This process can be completed by the help of BS. For example, if transmitter $i$ wants to know CSI $h_{i,i}$, it transmits training symbols in broadcast manner. Receiver $i$ first estimates and sends $h_{i,i}$ to BS and then BS sends the information to transmitter $i$. Such an information acquisition process for establishing D2D connections is defined in 3GPP specifications.

Moreover, for the optimal differentiated pricing, the cross channel gains $\{h_{j,i}\}_{j\neq i}$ additionally need to be estimated at the receivers. Since the training symbols are transmitted in broadcast manner, other adjacent D2D receivers $j\neq i$ also can receive them and thus estimate $h_{j,i}$ and then feed back to BS for computing differentiated prices. Note that $\{h_{j,i}\}_{j\neq i}$ are no need for the transmitters.

For the suboptimal differentiated pricing, the needed network information at the user side are the same as above, and the only difference is that the power allocation \eqref{eqn:subopt-p} does not need to iterate but the best response function  $\mathcal B_i(p_{-i})$ in \eqref{eqn:best} needs iteration.

\subsubsection{Complexity at BS Side}

For the uniform pricing at the BS side, the needed information are (please see \eqref{eqn:pi-u} and \eqref{eqn:pi-l}) channel gains  $h_{i,i}$ and $g_i$, IpN $\Delta_i(p_{-i})$, weight $w_i$ and peak power constraint $\overline{p}_i$. For the optimal differentiated pricing at the BS side, the needed information are (please see \eqref{eqn:transfer}) channel gains  $h_{i,i}$, $h_{j,i}$ ($j\neq i$) and $g_i$, weight $w_i$ and peak power constraint $\overline{p}_i$. Comparing the needed network information of the uniform and the optimal differentiated pricing, the only difference is that the uniform pricing needs IpN $\Delta_i(p_{-i})$ and the optimal differentiated pricing needs interference channel gains among users $h_{j,i}$ ($j\neq i$). The dimensions of $\{\Delta_i(p_{-i})\}$ and $\{h_{j,i}\}_{j\neq i}$ are $N\times1$ and $N\times (N-1)$, respectively, where $N$ is the number of users.

For the suboptimal differentiated pricing, the only difference compared to the optimal differentiated pricing is that the cross channel gains $\{h_{j,i}\}_{j\neq i}$ are no need by the assumption of the negligible interference among users.

\section{Simulation Results}

In this section, we conduct comprehensive simulations to evaluate the performance of the proposed algorithms for distributed power allocation and interference management in D2D cellular networks. Without loss of generality, we let the weights $w_i=1$ and the noise powers $\sigma^2=1$. We assume all D2D users have the same maximum power constraints in dB.
Since we focus on controlling the interference from the D2D users to the BS, we assume that in simulation there are only D2D transmissions for simplicity. Note that this does not affect the proposed algorithms since the impacts of uplink cellular transmissions can be integrated into the noise power in the SINR expressions of D2D users.

We consider a cell with a radius of $100$. For an illustration purpose, we consider  $N=4$ D2D users that are randomly but uniformly distributed within the coverage of the cell. The source-destination distance of each D2D user is randomly distributed between $(0,10]$.  The fading channels are modeled as $c\cdot L^{-\theta}$, where $c$ is the small-scale fading factor which is modeled by Rayleigh fading process, $L$ is the transmission distance and $\theta$ is the path loss exponent which is set to be $2$ for the large-scale fading.


%
\begin{figure*}[t]
\begin{centering}
\makeatletter\def\@captype{figure}\makeatother
\subfigure[$\bs p^{(0)}=\bs0$.]{\includegraphics[width=3.2in]{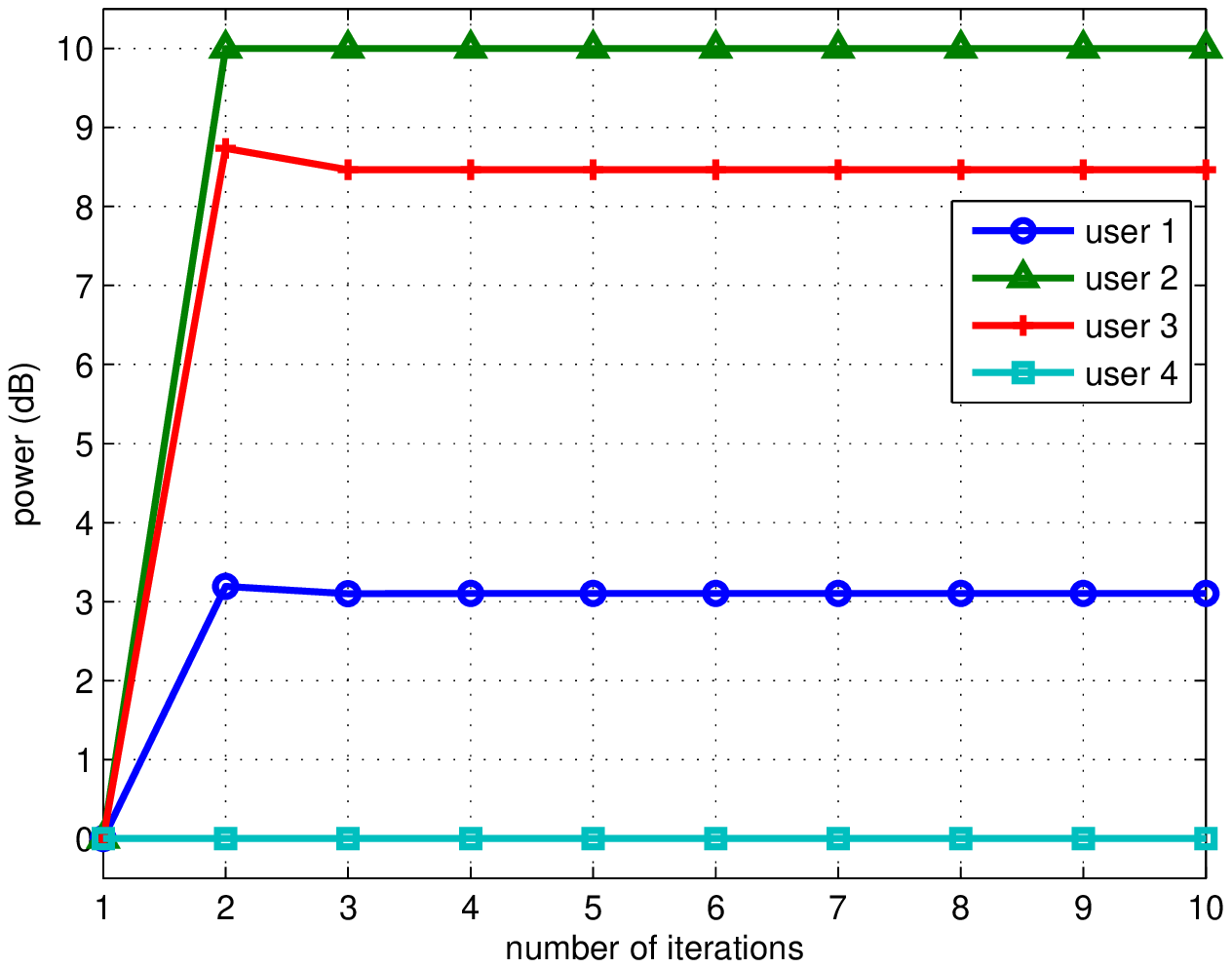}
\label{fig:example_revenue}}
\subfigure[$\bs p^{(0)}=\overline{\bs p}$.]{\includegraphics[width=3.2in]{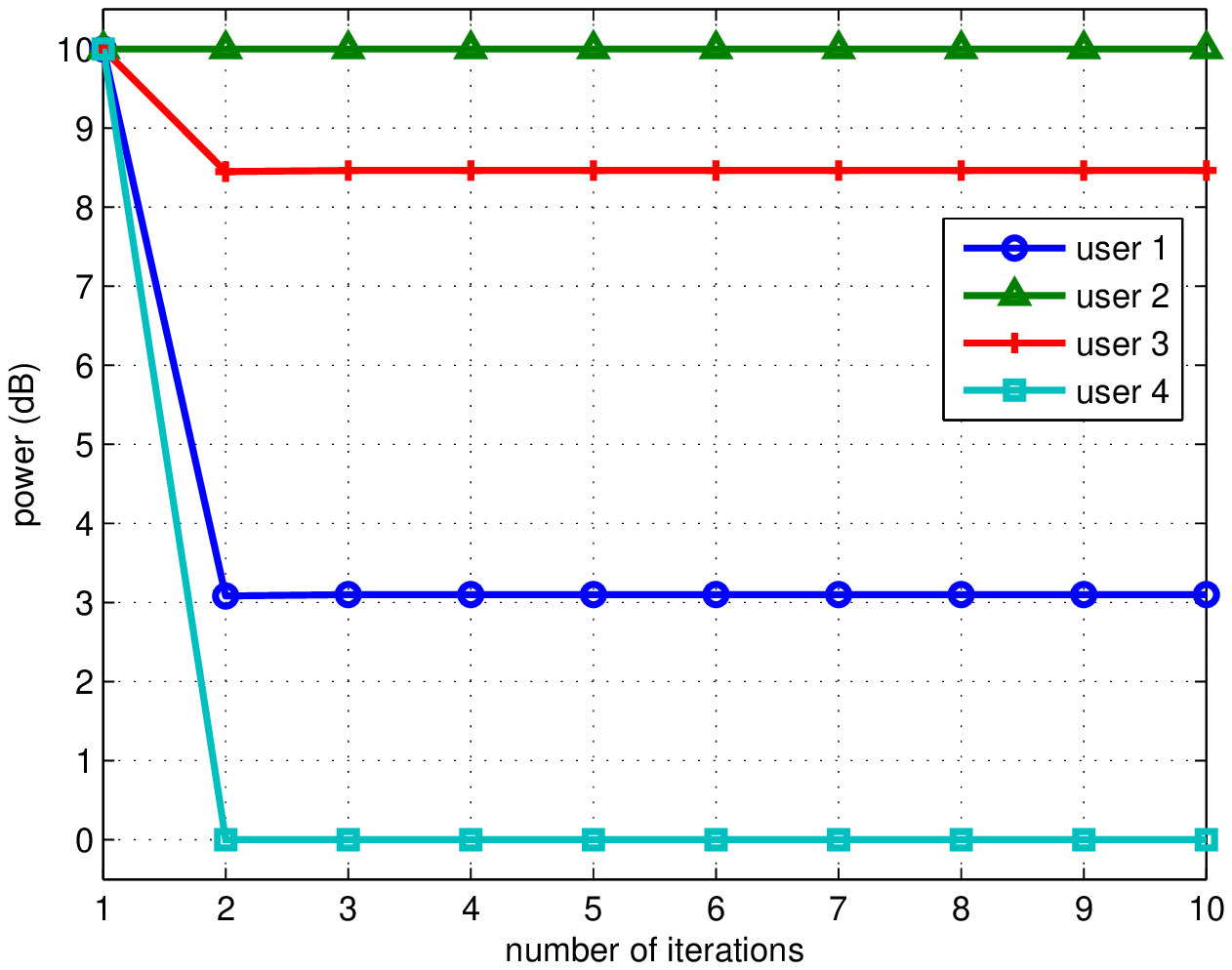}
\label{fig:example_power}}
\vspace{-0.1cm} \caption{Convergence performance of Algorithm 1 with 4 D2D users, where $\overline{p}_i=10$ dB, $\forall i$, and $\pi=\pi^u/10$.} \label{fig:4powers}
\end{centering}
\vspace{-0.3cm}
\end{figure*}

For a given random channel realization, we first investigate the convergence performance of the proposed iterative based distributed power allocation algorithm in Algorithm 1. As shown in Fig. \ref{fig:4powers}, where the peak power constraints are $\overline{p}_i=10$ dB, $\forall i$, and the price is prefixed as $\pi^u/10$, we can observe that the convergence speed of the proposed Algorithm 1 is very fast, in specific, only about $3$ iterations are needed for this example. Moreover, we observe that the initial power vectors $\bs p^{(0)}=\bs0$ and $\bs p^{(0)}=\overline{\bs p}$ do not affect the final power outputs, which verifies the Proposition \ref{prop:ne} that the NE point is regardless of the initial feasible power vector $\bs p^{(0)}$.

\begin{figure*}[t]
\begin{centering}
\makeatletter\def\@captype{figure}\makeatother
\subfigure[Revenue.]{\includegraphics[width=2.2in]{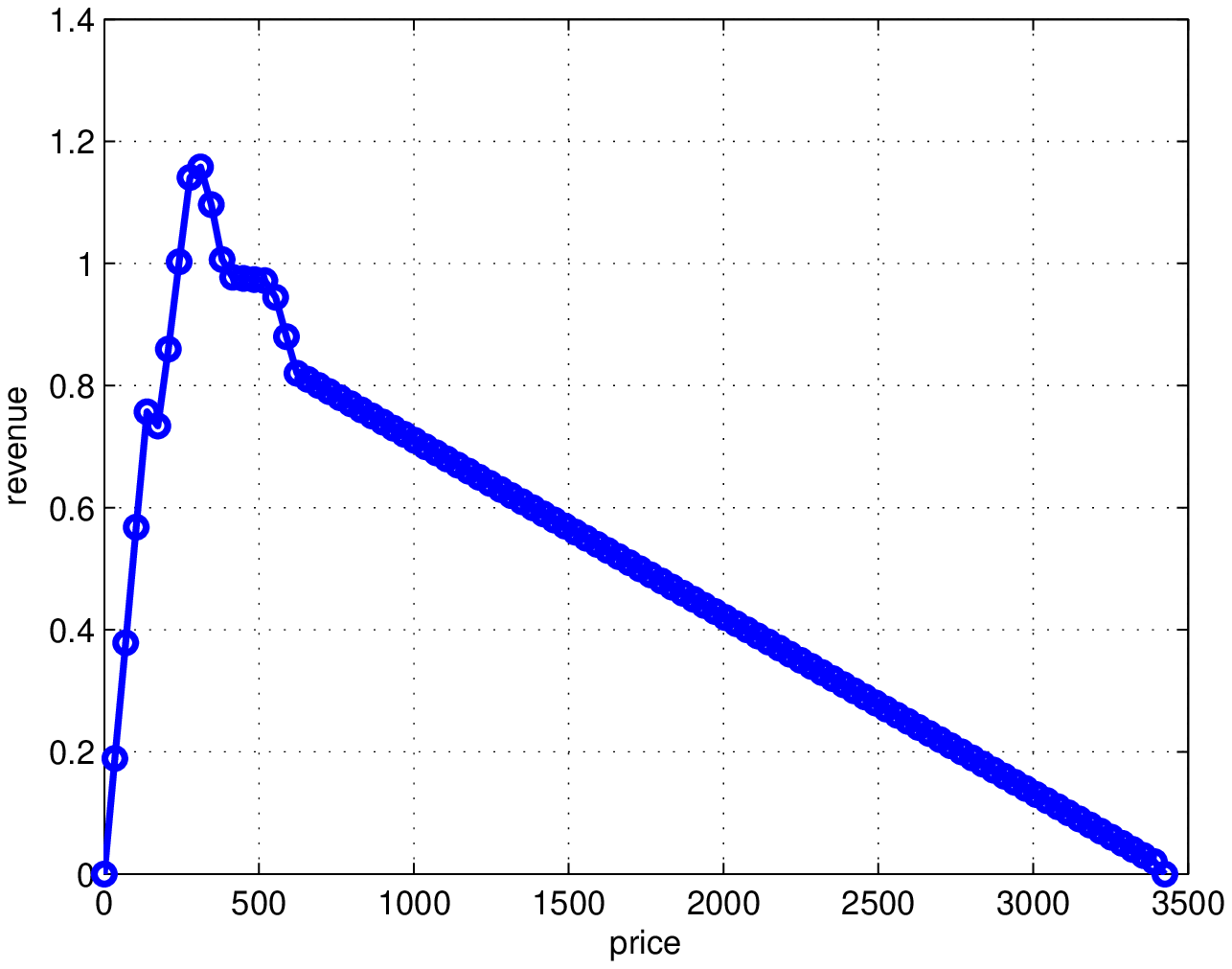}
\label{fig:example_revenue}}
\subfigure[Powers.]{\includegraphics[width=2.2in]{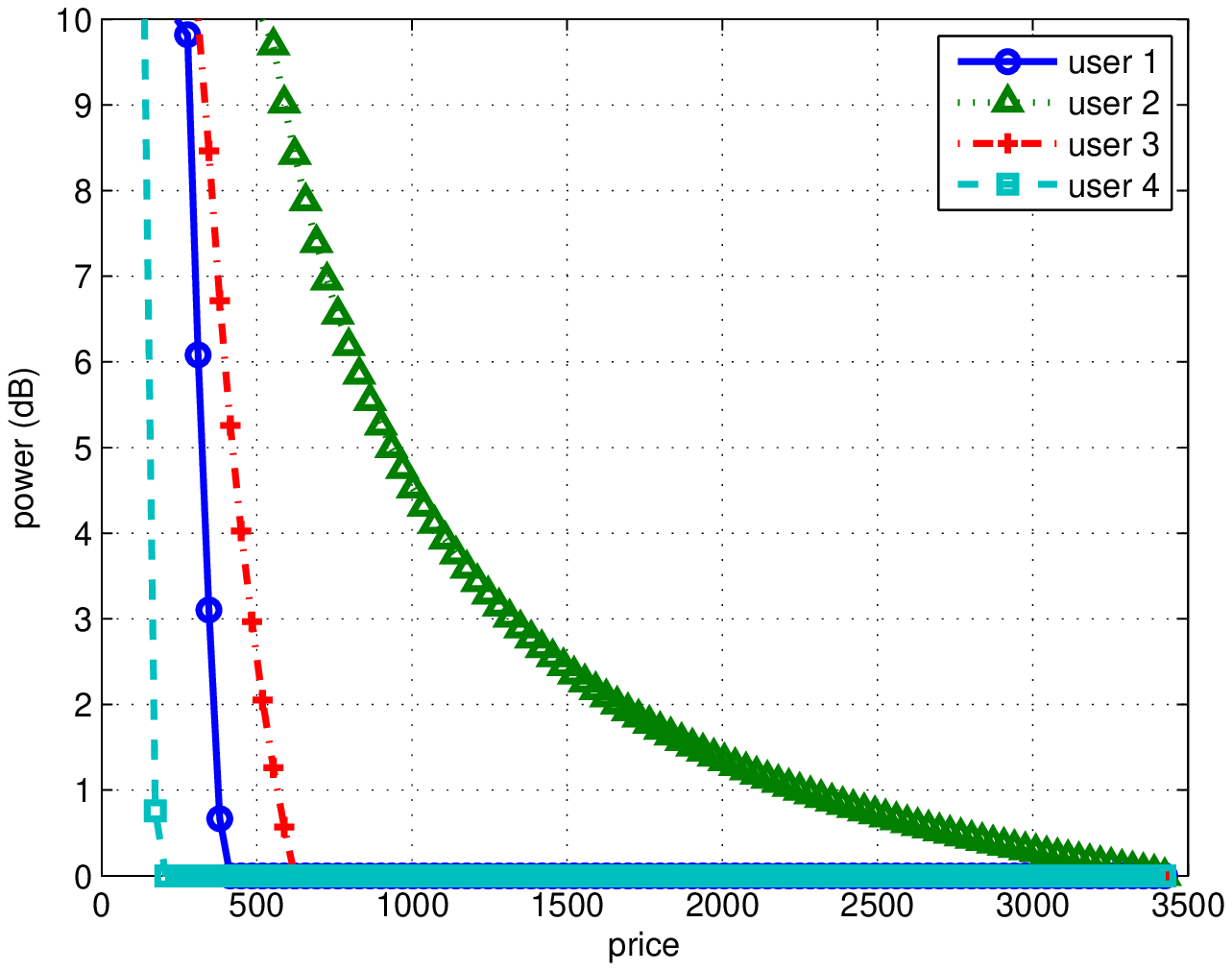}
\label{fig:example_power}}
\subfigure[Total interference at BS.]{\includegraphics[width=2.2in]{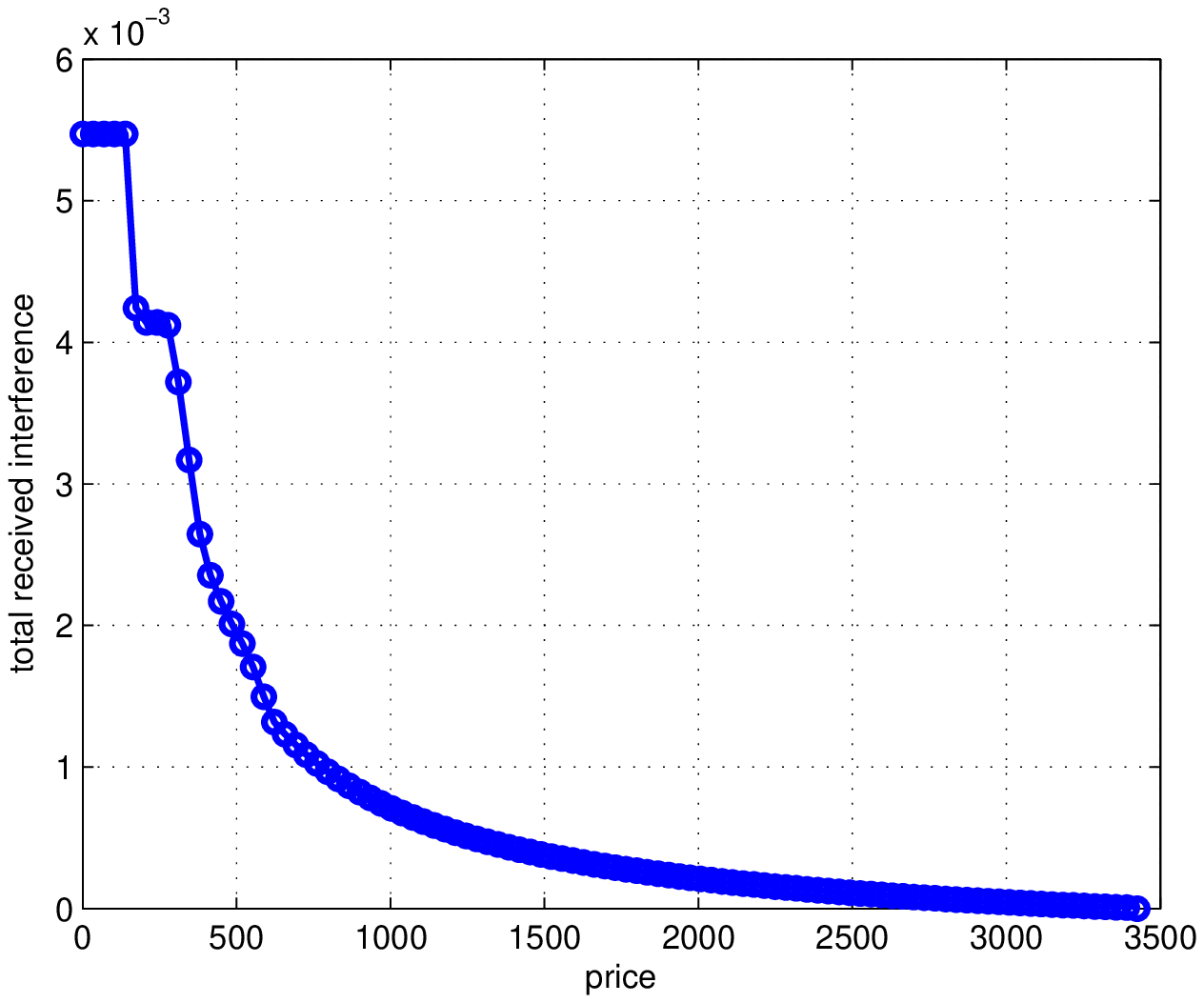}
\label{fig:example_interference}}
\vspace{-0.1cm} \caption{Performance of the uniform pricing scheme in Algorithm 2. } \label{fig:alg2}
\end{centering}
\vspace{-0.3cm}
\end{figure*}

Next, for the same channel realization, Fig. \ref{fig:alg2} studies the performance of the proposed uniform pricing scheme in Algorithm 2, where the peak power constraints are $\overline{p}_i=10$ dB, $\forall i$. In Fig. \ref{fig:example_revenue}, we observe that the revenue is linear with the price at the start and nonconvex after a certain point (i.e., lower bound price $\pi^l$) and finally becomes zero at a certain point (i.e., upper bound price $\pi^u$). One also observes from Fig. \ref{fig:example_power} that the powers keep maximum at the start (i.e., $[0,\pi^l]$) and are decreasing with the price, so is the aggregate interference at the BS in Fig. \ref{fig:example_interference}. These observations are in accordance with our analysis given in Section III-B.

\begin{figure}[t]
\begin{centering}
\includegraphics[scale=0.65]{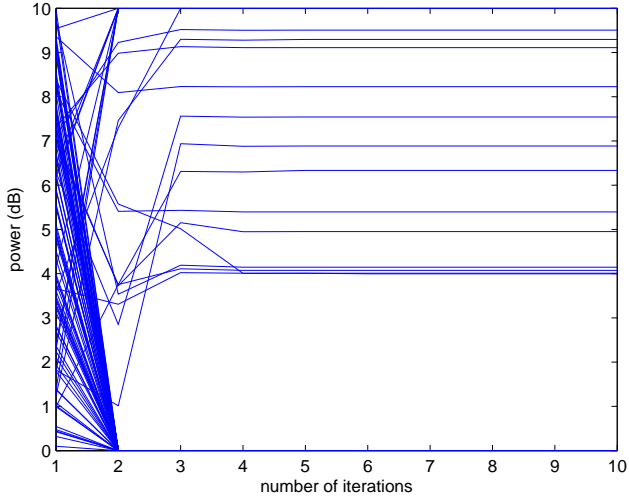}
\vspace{-0.1cm}
 \caption{Convergence performance of Algorithm 1 with 100 D2D pairs, where $\overline{p}_i=10$ dB, $\forall i$, and $\pi=\pi^u/10$.}\label{fig:100powers}
\end{centering}
\vspace{-0.3cm}
\end{figure}

Then we investigate the convergence performance of Algorithm 1 with 100 D2D users in Fig. \ref{fig:100powers} for a given channel realization, where we set $\overline{p}_i=10$ dB, $\forall i$, and $\pi=\pi^u/10$. We can observe that Algorithm 1 converges fast even with 100 users, i.e. 4 iterations in this example. This demonstrates the effectiveness of Algorithm 1.

\begin{figure}[t]
\begin{centering}
\includegraphics[scale=0.65]{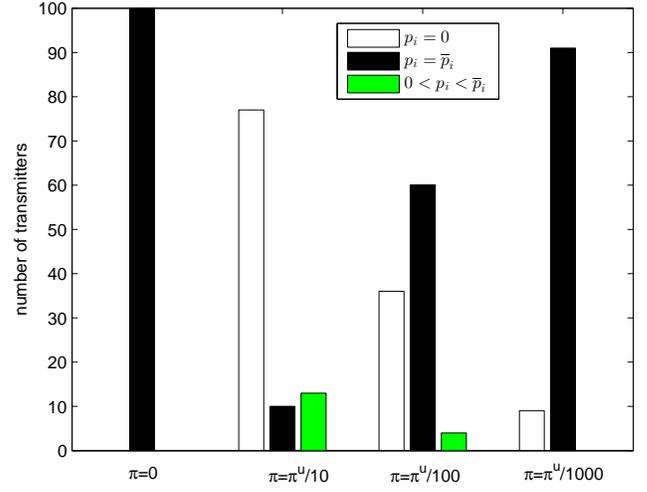}
\vspace{-0.1cm}
 \caption{The number of D2D transmitters of Algorithm 1 with different prices, where $N=100$ D2D pairs and $\overline{p}_i=10$ dB.}\label{fig:100num}
\end{centering}
\vspace{-0.3cm}
\end{figure}

Moreover, we investigate impact of the uniform price in Algorithm 1 on users' transmit power in Fig. \ref{fig:100num} for a given channel realization, where we set $\overline{p}_i=10$ dB, $\forall i$, and $N=100$ D2D pairs are considered. We can observe that all users will transmit their maximum power when $\pi=0$, which is accord with the conclusion in \cite{Rasti2009} that it is a regular noncooperative power control game in this case. We also observe that if the price is becoming larger, more and more users are not willing to transmit (i.e., $p_i=0$). This also coincides with our analysis in Section III.

In Figs. \ref{fig:snr} and \ref{fig:th}, we evaluate and compare the statistical (or average) performance of the three proposed pricing algorithms based on two distinct performance metrics, i.e., sum rates and revenue. A total of $1000$ channel realizations are used. For each channel realization, the location and the source-destination distance of each D2D user are random.  For the uniform pricing scheme in Algorithm 2, the step size of the price is selected as $\epsilon=(\pi^u-\pi^l)/1000$.
%

%
\begin{figure}[t]
\begin{centering}
\makeatletter\def\@captype{figure}\makeatother
\subfigure[Sum rates.]{\includegraphics[width=3.5in]{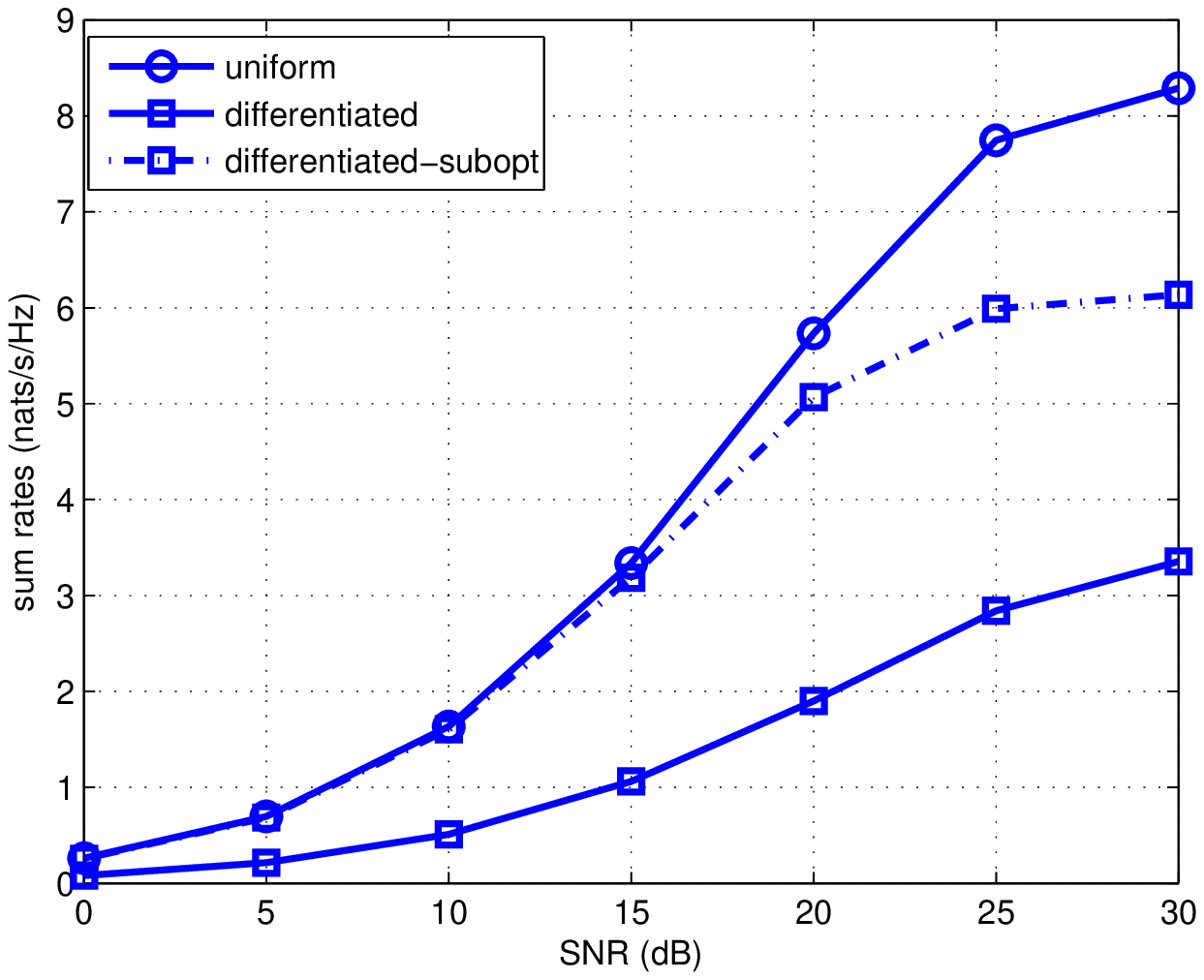}
\label{fig:4user_rate_snr}}
\subfigure[Revenue.]{\includegraphics[width=3.5in]{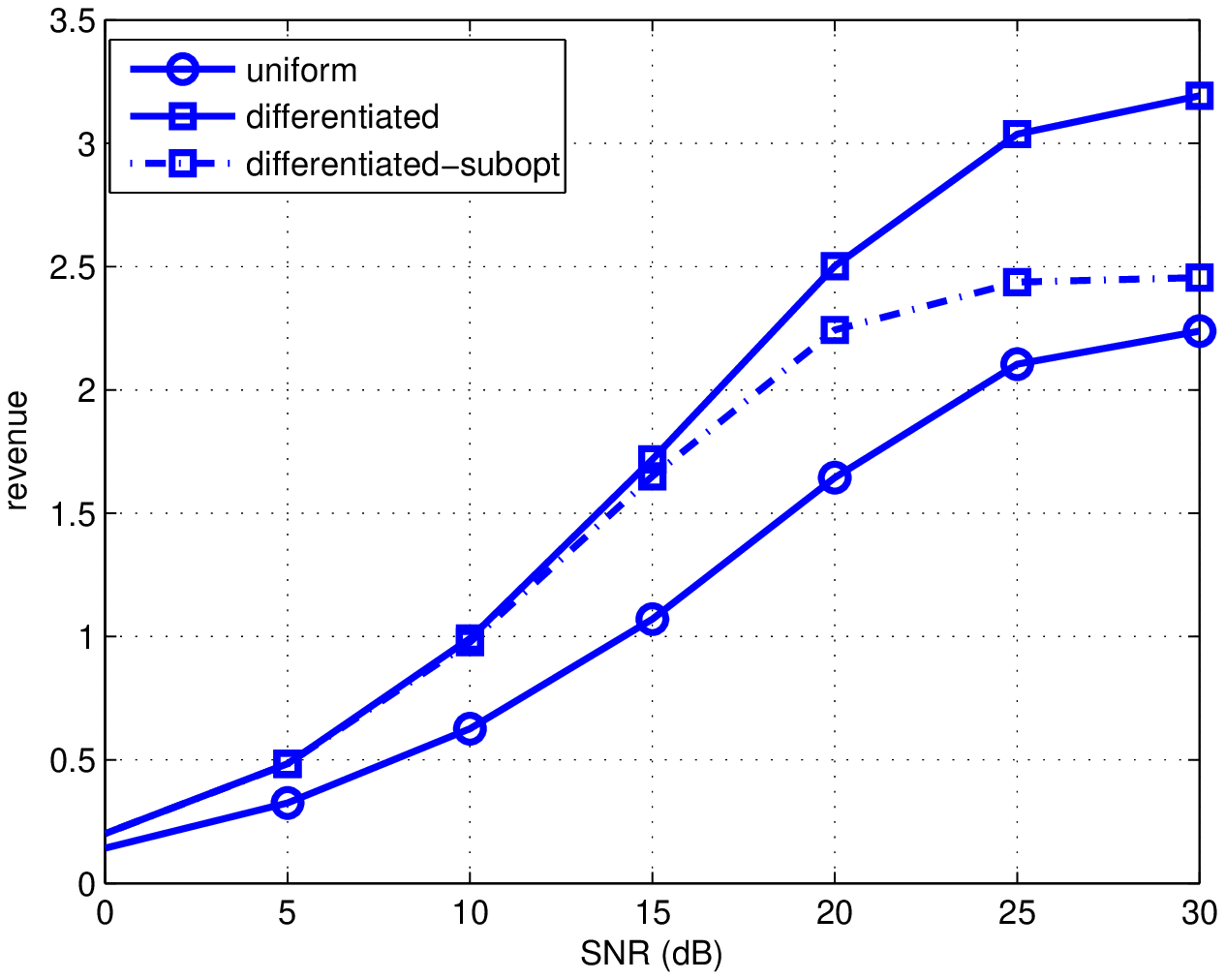}
\label{fig:4user_revenue_snr}}
\vspace{-0.1cm} \caption{Performance of the three proposed pricing algorithms, where the interference temperature constraint is set as $I_{th}=0.05$. } \label{fig:snr}
\end{centering}
\vspace{-0.3cm}
\end{figure}
\begin{figure}[t]
\begin{centering}
\makeatletter\def\@captype{figure}\makeatother
\subfigure[Sum rates.]{\includegraphics[width=3.5in]{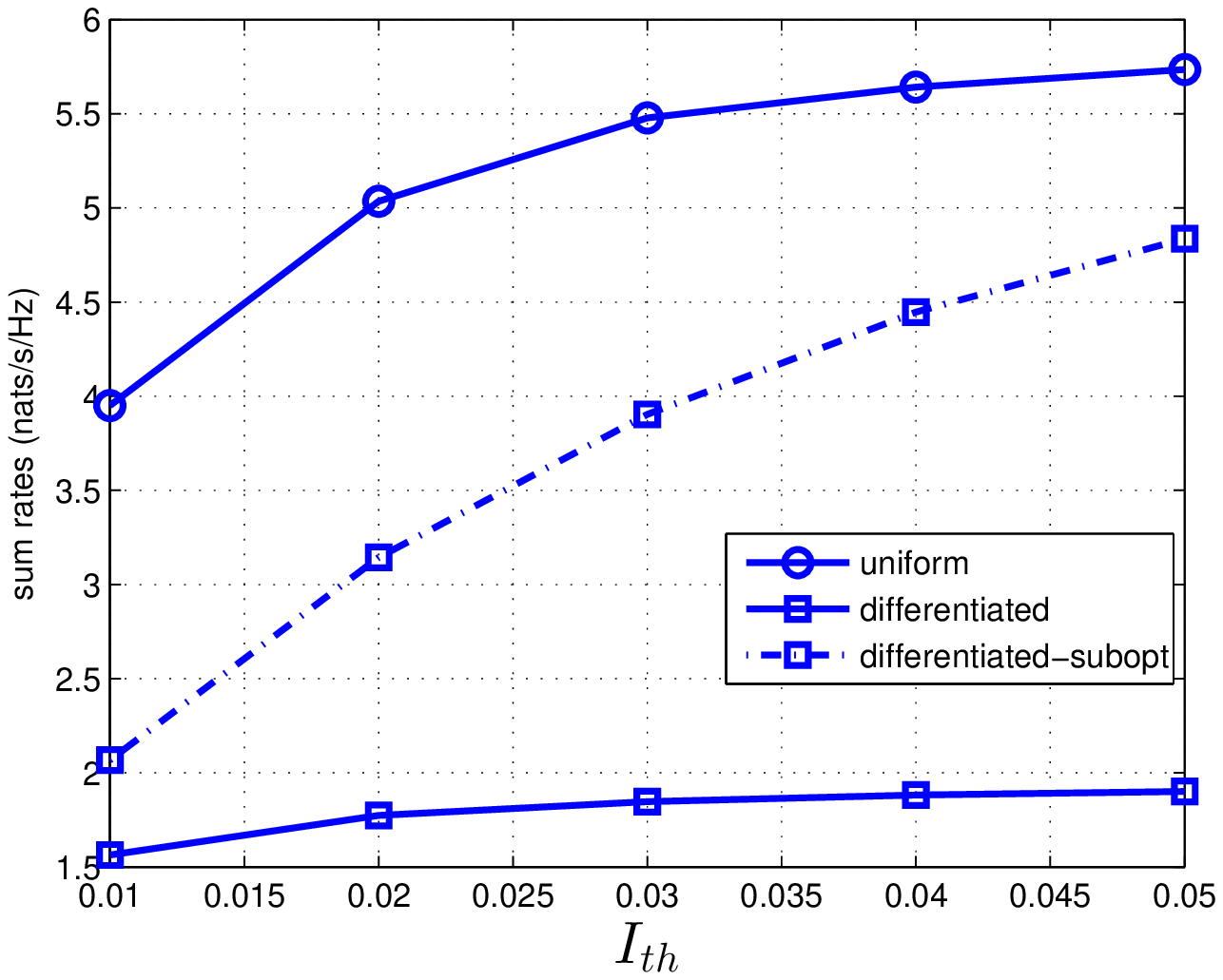}
\label{fig:4user_rate_th}}
\subfigure[Revenue.]{\includegraphics[width=3.5in]{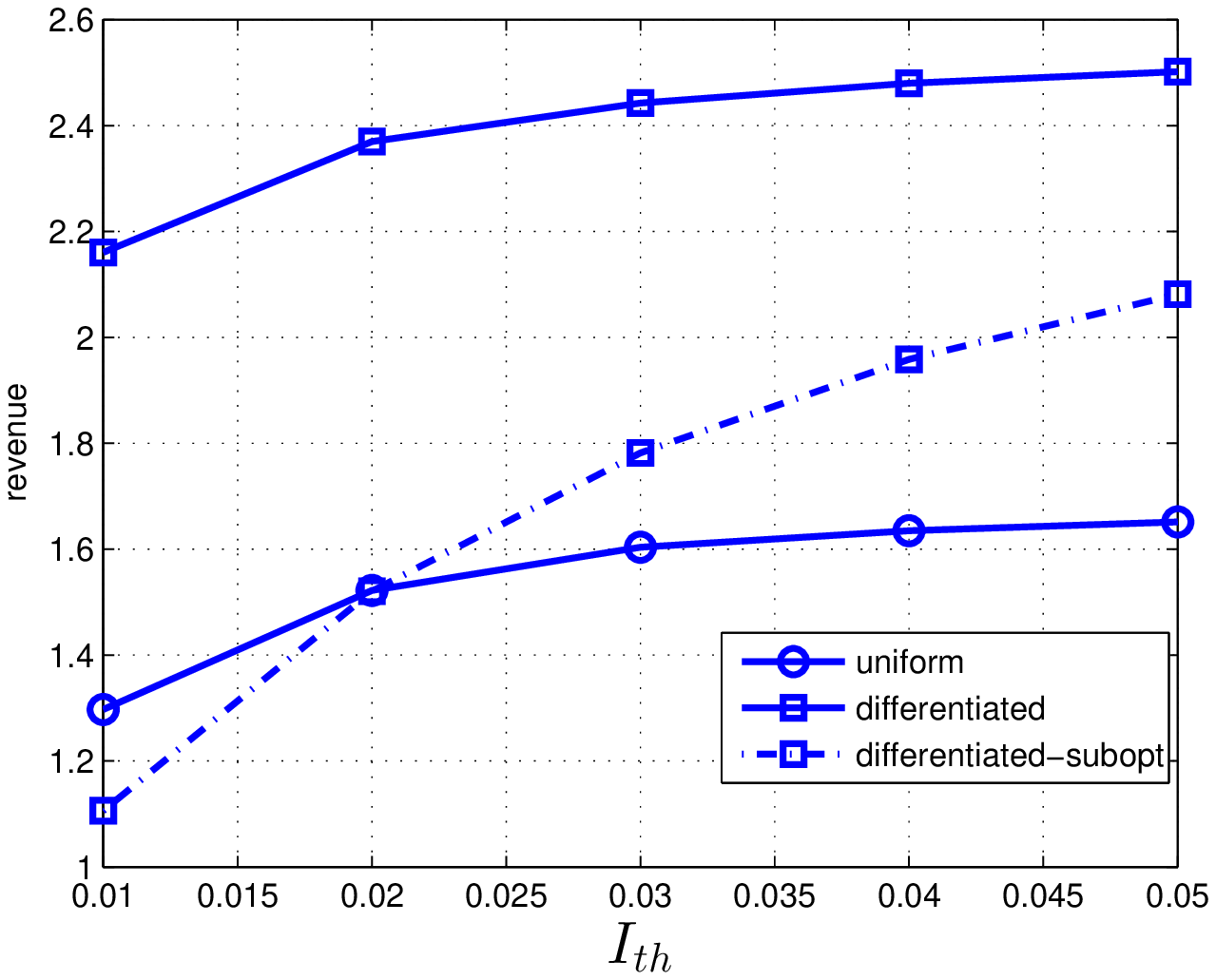}
\label{fig:4user_revenue_th}}
\vspace{-0.1cm} \caption{Performance of the three proposed pricing algorithms, where the peak power constraints are set as $\overline{p}_i=20$ dB, $\forall i$. } \label{fig:th}
\end{centering}
\vspace{-0.3cm}
\end{figure}

In Fig. \ref{fig:snr}, we fix the interference temperature constraint $I_{th}=0.05$.
It is observed that the revenue of the BS with the differentiated pricing schemes is generally larger than that with the uniform pricing scheme, while it is reverse for the sum-rate of the D2D users.
Note that the D2D users and the BS have the conflicting objectives. That is, higher utility of the D2D users means lower utility of the BS, and vice versa.  Hence, more efficient pricing scheme yields to higher utility of the BS but lower utility of the D2D users.
We observe from Fig. \ref{fig:4user_rate_snr} that, when SNR is below $15$ dB, the suboptimal differentiated pricing scheme and the uniform pricing scheme perform closely in terms of sum rates of the D2D users, but the uniform pricing scheme has better performance if SNR is higher than $15$ dB. Fig. \ref{fig:4user_revenue_snr} shows that the two differentiated pricing schemes outperform the uniform pricing scheme in terms of revenue over a wide range of SNR. Moreover, the two differentiated pricing schemes have the similar revenue performance when SNR is below $15$ dB, and the optimal differentiated pricing scheme outperforms the suboptimal one when SNR is higher than $15$ dB.

In Fig. \ref{fig:th}, we fix the peak power constraints $\overline{p}_i=20$ dB, $\forall i$. For the performance of sum rates, Fig. \ref{fig:4user_rate_th} shows that the uniform pricing scheme outperforms the two differentiated pricing schemes, and the suboptimal differentiated pricing scheme is better than the optimal differentiated pricing scheme in term of sum-rate of the users. The reason is mentioned in above. For the performance of revenue, Fig. \ref{fig:4user_revenue_th} illustrates that the optimal differentiated pricing scheme has the best performance as expected. The uniform pricing scheme is better than the suboptimal differentiated pricing scheme when $I_{th}$ is stringent, and the uniform pricing scheme is worse than the suboptimal differentiated pricing scheme when $I_{th}$ is loose.


\section{Conclusion}

In this paper, we studied distributed power allocation and interference management for the D2D enabled cellular networks. The interference temperature constraint was applied at the BS to ensure that the aggregate received interference from the D2D users is below a threshold. The interactions between the BS and D2D users were modeled as a Stackelberg game. Based on the noncooperative game theory, we proposed an iterative distributed power allocation algorithm for the D2D users, and proved that there always exists a unique NE point. We proposed a uniform pricing algorithm for the interference-constrained revenue maximization at the BS, under the assumption that the BS only knows the limited network information. Then a differentiated pricing algorithm was also presented by assuming that the BS has the global network information. We also proposed a suboptimal differentiated pricing scheme to reduce complexity. It was shown that the uniform pricing algorithm can be implemented with distributed manner and requires minimum information exchange between the BS and the D2D users. We also showed that the optimal and suboptimal differentiated pricing algorithms are iteration-free between the BS and the D2D users, and they are partially and fully distributed, respectively. The proposed framework and algorithms are useful and practical for resource allocation and interference management in spectrum-sharing D2D communications.

There are several research directions for future work. In this paper, the network nodes are assumed to be cooperative, which means that the BS and users follow the proposed algorithms, including the BS announcing the correct price signals and the users transmitting their pilots/beacons at the correct power levels. A challenging extension is to provide incentives for players to report the correct signals truthfully, or design punishment policies against cheating behaviors. Moreover, the nodes are assumed as myopic and the solution is NE, it is also interesting to consider the foresight players such that the nodes aim to maximize their long-term payoffs instead of immediate payoffs.

\appendices
\section{Proof of Proposition \ref{prop:ne}}\label{app:ne}
Here we only describe the key points and briefly to prove Proposition \ref{prop:ne} as follows.

To prove the existence of NE, we can show that: in subgame $\mathcal G$, for every user, the power strategy space $\mathcal P_i$ is a nonempty, convex, and compact subset of a Euclidean space, and the payoff function $u_i(p_i, p_{-i}, \pi_i)$ is continuous in $\bs p$ and quasi-concave\footnote{Quasi-concave is a generalization of concave.} in $p_i$, for all $i\in\mathcal N$ \cite{Fudenberg}.

Next, we prove the uniqueness of NE. By definition, the NE is the fixed point in the best response function set that satisfies $\bs p= \mathcal B(\bs p)$. For the two extreme cases of $\mathcal B(\bs p)=\bs0$ and $\mathcal B(\bs p)=\overline{\bs p}$, where $\overline{\bs p}=\{\overline{p}_1,\cdots,\overline{p}_N\}$, the fixed point of the best response function is unique and corresponds to the peak transmit power and zero transmit power for all users, respectively. For the other cases, we prove the uniqueness relying on the concept of \emph{standard function} \cite{Yates}. It has been proven in \cite{Yates} that the NE point (if it exists) in a standard function is unique. A function $\mathcal B(\bs p)$ is said to be standard if for all feasible power $\bs p$ , the following conditions hold \cite{Yates}
\begin{itemize}
  \item Positivity: $\mathcal B(\bs p)>0$;
  \item Monotonicity: if $\bs p\succeq \tilde{\bs p}$, then $\mathcal B(\bs p)\succeq\mathcal B(\tilde{\bs p})$;
  \item Scalability: for all $c>1$, $c\mathcal B(\bs p)>\mathcal B(c\bs p)$.
\end{itemize}
Here we omit the routine details of the proof.
%

It is worth noting that the existence of a fixed point (even it is unique) of an iterative process does not necessarily maintain the convergence, and the existence of a fixed point and convergence are two separate concepts of an iterative process. We prove that the proposed iterative distributed algorithm in Algorithm 1 can converge to the unique NE since the best response function is standard and each user has a peak power constraint \cite{Yates}.

\section{Proof of Proposition \ref{prop:ub}}\label{app:ub}
Properties 1) and 2) can be easily observed. Let $\mathcal B(\bs p)=\overline{\bs p}$ and $\mathcal B(\bs p)=\bs0$, we can obtain $\pi^l$ and $\pi^u$, respectively.
If
\begin{equation}
\pi \geq \pi^u\triangleq\max_{i\in\mathcal N}\frac{w_ih_{i,i}}{g_i\sigma^2},
\end{equation}
then 
\begin{equation}
\pi\geq\frac{w_ih_{i,i}}{g_i\sigma^2},~\forall i,
\end{equation}
which means that 
\begin{equation}
\frac{w_i}{g_i\pi}\leq\frac{\sigma^2}{h_{i,i}},~\forall i.
\end{equation}
Combining the fact that
\begin{equation}
\frac{\sigma^2}{h_{i,i}}\leq\frac{\sigma^2+\sum_{j\neq i}p_{j}h_{j,i}}{h_{i,i}}\triangleq\frac{\Delta_i(p_{-i})}{h_{i,i}},
\end{equation}
we conclude that 
\begin{equation}
\frac{w_i}{g_i\pi}\leq\frac{\Delta_i(p_{-i})}{h_{i,i}},~\forall i.
\end{equation}
Plugging it into the best response function $\mathcal B_i(p_{-i})$,
it is concluded that $\mathcal B_i(p_{-i})=0$ for all $i\in\mathcal N$ and thus $u_B(\pi)=0$. Moreover, $u_B(\pi) = 0$ if $\pi = 0$ obviously holds.

Similarly,  if 
\begin{equation}
\pi\leq\pi^l\triangleq\min_{i\in\mathcal N}\frac{w_ih_{i,i}}{g_i\left(\overline{p}_ih_{i,i}+\Delta_i(\overline{p}_{-i})\right)},
\end{equation}
then 
\begin{equation}
\pi\leq\frac{w_ih_{i,i}}{g_i\left(\overline{p}_ih_{i,i}+\Delta_i(\overline{p}_{-i})\right)},~\forall i,
\end{equation}
which means that 
\begin{equation}
\frac{w_i}{g_i\pi}\geq\frac{\overline{p}_ih_{i,i}+\Delta_i(\overline{p}_{-i})}{h_{i,i}}=
\overline{p}_i+\frac{\Delta_i(\overline{p}_{-i})}{h_{i,i}},
\end{equation}
or equivalently 
\begin{equation}
\frac{w_i}{g_i\pi}-\frac{\Delta_i(\overline{p}_{-i})}{h_{i,i}}\geq\overline{p}_i.
\end{equation}
Plugging it into the best response function $\mathcal B_i(p_{-i})$,
it is concluded that $\mathcal B_i(p_{-i})=\overline{p}_i$ for all $i\in\mathcal N$ and thus $u_B(\pi)=\pi\sum_{i=1}^N\overline{p}_i g_i$.

On the other hand, if $u_B(\pi)=0$, then $\pi=0$ or $p_i=0$ for all $i$. The former case is trivial and for the later case, it should be 
\begin{equation}
\frac{w_i}{g_i\pi}-\frac{\Delta_i(p_{-i})}{h_{i,i}}\leq0,~\forall i.
\end{equation}
This leads to 
\begin{equation}
\pi\geq\frac{w_ih_{i,i}}{g_i\sigma^2},~\forall i,
\end{equation}
which means 
\begin{equation}
\pi\geq\max_{i\in\mathcal N}\frac{w_ih_{i,i}}{g_i\sigma^2}\triangleq\pi^u.
\end{equation}

Similarly, if $u_B(\pi)=\pi\sum_{i=1}^N\overline{p}_i g_i$, it should be $p_i=\overline{p}_i$ for all $i$, or 
\begin{equation}
\frac{w_i}{g_i\pi}-\frac{\Delta_i(p_{-i})}{h_{i,i}}\geq\overline{p}_i,~\forall i.
\end{equation}
This means that 
\begin{equation}
\min_{i\in\mathcal N}\frac{w_i}{g_i\pi}-\frac{\Delta_i(p_{-i})}{h_{i,i}}\geq\overline{p}_i,
\end{equation}
and thus 
\begin{equation}
0\leq\pi\leq\min_{i\in\mathcal N}\frac{w_ih_{i,i}}{g_i\left(\overline{p}_ih_{i,i}+\Delta_i(\overline{p}_{-i})\right)}\triangleq\pi^l.
\end{equation}

The two extreme cases $\mathcal B(\bs p)=\bs0$ and $\mathcal B(\bs p)=\overline{\bs p}$ are also the NE points for given price. Then the properties 3) and 4) can be proved.

\section{Proof of Corollary \ref{coro:pi}}\label{app:pi}
 By differentiating the best response function $\mathcal B_i(p_{-i})$ with respect to $\pi$, it can be shown that $\mathcal B_i(p_{-i})$ is a strictly decreasing function of $\pi$ when $\pi^l\leq\pi\leq\pi^u$. As stated in the proof of Proposition \ref{prop:ub}, $\mathcal B_i(p_{-i})=\overline{p}_i$ when $\pi=\pi^l$. Therefore the peak power constraint of user $i$ is not active if $\pi^l<\pi\leq\pi^u$, i.e., $0\leq\mathcal B_i(p_{-i})<\overline{p}_i$, $\forall i\in\mathcal N$. When $\pi>\pi^u$, $\mathcal B_i(p_{-i})\equiv0$, $\forall i\in\mathcal N$.
This completes the proof.

\bibliographystyle{IEEEtran}
\bibliography{IEEEabrv,price}

\end{document}